\documentclass[preprint,showpacs,preprintnumbers,amsmath,amssymb]{revtex4}
\usepackage{amsmath}

\usepackage{graphicx}
\usepackage{dcolumn}
\usepackage{bm}
\usepackage{amsmath}

\newtheorem{theorem}{Theorem}[section]
\newtheorem{lemma}[theorem]{Lemma}
\newtheorem{proposition}[theorem]{Proposition}
\newtheorem{corollary}[theorem]{Corollary}

\newenvironment{proof}[1][Proof]{\begin{trivlist}
\item[\hskip \labelsep {\bfseries #1}]}{\end{trivlist}}
\newenvironment{definition}[1][Definition]{\begin{trivlist}
\item[\hskip \labelsep {\bfseries #1}]}{\end{trivlist}}

\newcommand{\qed}{\nobreak \ifvmode \relax \else
      \ifdim\lastskip<1.5em \hskip-\lastskip
      \hskip1.5em plus0em minus0.5em \fi \nobreak
      \vrule height0.75em width0.5em depth0.25em\fi}

\begin{document}

\preprint{}

\title{Generalized Two-Qubit Whole and Half Hilbert-Schmidt Separability Probabilities} 
\author{Paul B. Slater}
 \email{slater@kitp.ucsb.edu}
\affiliation{%
University of California, Santa Barbara, CA 93106-4030\\
}
\author{Charles F. Dunkl}
 \email{cfd5z@virginia.edu}
\affiliation{Department of Mathematics, University of Virginia,
Charlottesville, VA 22904-4137}%
\date{\today}

\begin{abstract}
Compelling evidence--though yet no formal proof--has been adduced that the probability that a generic (standard) 
two-qubit state ($\rho$) is separable/disentangled is 
$\frac{8}{33}$ (arXiv:1301.6617, arXiv:1109.2560, arXiv:0704.3723). Proceeding in related analytical frameworks, using a further determinantal $4F3$-hypergeometric moment formula  (Appendix A), we reach, {\it via} density-approximation (inverse) procedures, the conclusion that one-half ($\frac{4}{33}$) of this probability arises when the determinantal inequality $|\rho^{PT}|>|\rho|$, where $PT$ denotes the partial transpose, is satisfied, and, the other half,  when $|\rho|>|\rho^{PT}|$. These probabilities are taken with respect to the flat, Hilbert-Schmidt measure on the fifteen-dimensional convex set of $4 \times 4$ density matrices. We find fully parallel bisection/equipartition results for the previously adduced, as well,  two-``re[al]bit'' and two``quater[nionic]bit'' separability probabilities of $\frac{29}{64}$ and 
$\frac{26}{323}$, respectively. The new determinantal $4F3$-hypergeometric moment formula is, then, adjusted (Appendices B and C) to the boundary case of minimally degenerate states 
($|\rho|=0$), and its consistency manifested--also using density-approximation--with an important theorem of Szarek, Bengtsson and {\.Z}yczkowski (arXiv:quant-ph/0509008). This theorem states that the Hilbert-Schmidt separability probabilities of generic minimally degenerate two-qubit states are (again) one-half those of the corresponding generic nondegenerate states.
\end{abstract}


\maketitle
\section{Introduction}
The problem of determining the probability that a bipartite/multipartite quantum state of  a certain random nature exhibits a particular entanglement characteristic  is clearly of intrinsic ``philosophical, practical, physical'' \cite{ZHSL}) interest \cite{ZHSL,simon,BHN,sbz,ingemarkarol,Bhosale,McNulty}. We have reported \cite{MomentBased,slaterJModPhys} major advances, in this regard, with respect to the ``separability/disentanglement probability'' of generalized two-qubit states (representable by $4 \times 4$ density matrices $\rho$), endowed with the flat, Hilbert-Schmidt  measure \cite{szHS,ingemarkarol}. 
Most noteworthy, a concise formula \cite[eqs. (1)-(3)]{slaterJModPhys} 
\begin{equation} \label{Hou1}
P(\alpha) =\Sigma_{i=0}^\infty f(\alpha+i),
\end{equation}
where
\begin{equation} \label{Hou2}
f(\alpha) = P(\alpha)-P(\alpha +1) = \frac{ q(\alpha) 2^{-4 \alpha -6} \Gamma{(3 \alpha +\frac{5}{2})} \Gamma{(5 \alpha +2})}{3 \Gamma{(\alpha +1)} \Gamma{(2 \alpha +3)} 
\Gamma{(5 \alpha +\frac{13}{2})}},
\end{equation}
and
\begin{equation} \label{Hou3}
q(\alpha) = 185000 \alpha ^5+779750 \alpha ^4+1289125 \alpha ^3+1042015 \alpha ^2+410694 \alpha +63000 = 
\end{equation}
\begin{displaymath}
\alpha  \bigg(5 \alpha  \Big(25 \alpha  \big(2 \alpha  (740 \alpha
   +3119)+10313\big)+208403\Big)+410694\bigg)+63000
\end{displaymath}
has emerged that yields for a given 
$\alpha$, where $\alpha$ is a random-matrix-Dyson-like-index \cite{dumitriu,caselle}, the corresponding Hilbert-Schmidt separability probability $P(\alpha)$. 
The setting $\alpha=1$ pertains to the fifteen-dimensional convex set of (standard/conventional, off-diagonal {\it complex}-entries) two-qubit density ($4 \times 4$ Hermitian, unit-trace, positive-semidefinite) matrices. 

The succinct formula yields (to arbitrarily high numerical precision) $P(1) =\frac{8}{33}$ (cf. \cite{slater833}, \cite[eq. B7]{joynt}, \cite[sec. VII]{Fonseca-Romero}). It is interesting to note that in this standard quantum-mechanical case \cite{steve}, the probability seems of a somewhat simpler nature (smaller numerators and denominators) than the value $P(\frac{1}{2}) =\frac{29}{64}$ obtained for the (``attractive toy model'' \cite{sbz}) nine-dimensional convex set of $4 \times 4$ (two-``rebit'') density matrices with {\it real} entries \cite{carl}, or, the value  $P(2) = \frac{26}{323}$ derived for the 
twenty-seven-dimensional convex set of $4 \times 4$ (two-``quaterbit'' \cite{batle2}) density matrices with {\it quaternionic} entries \cite{asher2,adler}.
(Let us note that $P(\frac{3}{2}) = \frac{36061}{262144}$ \cite[p. 9]{slaterJModPhys}. However, unlike the results for $\alpha = \frac{1}{2}, 1$ and 2, we have not been able to obtain this value through
direct density-matrix calculations. This disparity may be attributable to the proposition that the only associative real division algebras are the real numbers, complex numbers, and quaternions \cite{May}.)

Fei and Joynt \cite{FeiJoynt} have recently found strong support for these three primary conjectures by Monte Carlo sampling, using the extraordinarily large number of $5 \times 10^{11}$ points for each of the three cases  (cf. \cite[eq. (30)]{MiSt}, \cite{KhRo}).
\subsection{Multi-step derivation of concise formula} 
These simple rational-valued $\alpha$-parameterized separability probabilities and the formula $P(\alpha)$ above that yields them were obtained through a number of distinct steps of analysis. First, based on extensive computations (employing Cholesky matrix decompositions/parameterizations, Dirichlet measure and integration over spheres), we inferred the (yet formally unproven) determinantal-moment formula \cite[p. 30]{MomentBased} (cf. \cite[eq. (28)]{zozor} \cite{wan})
\begin{gather*} \label{nequalzero}
\left\langle \left\vert \rho^{PT}\right\vert ^{n}\right\rangle =\frac
{n!\left(  \alpha+1\right)  _{n}\left(  2\alpha+1\right)  _{n}}{2^{6n}\left(
3\alpha+\frac{3}{2}\right)  _{n}\left(  6\alpha+\frac{5}{2}\right)  _{2n}}\\
+\frac{\left(  -2n-1-5\alpha\right)  _{n}\left(  \alpha\right)  _{n}\left(
\alpha+\frac{1}{2}\right)  _{n}}{2^{4n}\left(  3\alpha+\frac{3}{2}\right)
_{n}\left(  6\alpha+\frac{5}{2}\right)  _{2n}}~_{5}F_{4}\left(
\genfrac{}{}{0pt}{}{-\frac{n-2}{2},-\frac{n-1}{2},-n,\alpha+1,2\alpha
+1}{1-n,n+2+5\alpha,1-n-\alpha,\frac{1}{2}-n-\alpha}%
;1\right) .
\end{gather*}
The brackets here denote expectation with respect to Hilbert-Schmidt (Euclidean) measure, while $5F4$ indicates a particular generalized hypergeometric function. The partial transpose of $\rho$, obtainable by transposing in place its four $2 \times 2$ blocks, is denoted by 
$\rho^{PT}$.

The first 7,501 of these moments ($n =0,1,\ldots 7500$) were employed as input to a Mathematica program of Provost \cite[pp. 19-20]{Provost}, implementing a Legendre-polynomial-based-density-approximation routine. From the high-precision, exact-arithmetic results obtained, we were able to formulate highly convincing, well-fitting conjectures (including the above-mentioned $\frac{8}{33}$ for $\alpha =1$) as to  underlying simple rational-valued separability probabilities. Then, with the use of the Mathematica FindSequenceFunction command applied to the sequence ($\alpha =1, 2,\ldots,32$)--or, fully equivalently, 
$\alpha =\frac{1}{2},\ldots \frac{63}{2}$--of these conjectures, and simplifying manipulations of the lengthy Mathematica result generated, we derived a multi-term $7F6$ hypergeometric-based formula \cite[Fig. 3]{slaterJModPhys} (cf. \cite[eq. (11)]{karolkarol}), with argument $\frac{27}{64}= (\frac{3}{4})^3$, for the conjectured values.  Then, Qing-Hu Hou (private communication) applied  a highly celebrated (``creative telescoping'') algorithm of Zeilberger \cite{doron} to this $7F6$-based expression to obtain the concise separability probability formula ((\ref{Hou1})-(\ref{Hou3})) for $P(\alpha)$ itself \cite[Figs. 5, 6]{slaterJModPhys}.
\subsection{General remarks}
Let us note that although the extensive symbolic and numeric computations conducted throughout this broad research project, have not furnished the rigorous proofs we, of course, strongly desire, they have been central to the testing of different approaches, and to the advancement of the specific 
determinantal-moment conjectures used for separability-probability evaluation. The conjectures take the form of equations asserted to hold for infinite ranges of parameter values, which can be verified for specific values of these parameters by symbolic computation.

Parallel programs to this one are being pursued in which: (1) the  theoretically-important Bures (minimal monotone) measure 
\cite{szBures,ingemarkarol,samcarl}--rather than the Hilbert-Schmidt one--is applied to the $4 \times 4$ density matrices; and (2) the $6 \times 6$ (qubit-qutrit) systems are studied with the Hilbert-Schmidt measure appropriate to them. Considerably less progress has so far been achieved in these areas.
No {\it general} moment formulas have yet been advanced, with explicit specific moment calculations having been implemented for the real and complex density matrices, so far for $\alpha = \frac{1}{2}$ and $\alpha =1$ for small values of 
$n$ \cite[sec. 6]{MomentBased} \cite{BuresHilbert,Hybrid}.
\subsection{Outline of study}
In sec.~\ref{Nondegenerate}, we change our previous focus in \cite{MomentBased,slaterJModPhys} from the moments and probability distributions associated with $|\rho^{PT}|$ to the associated variable $(|\rho^{PT}|-|\rho|)$, for which certain results appear to simplify, and in 
sec.~\ref{Degenerate} to the variable $|\rho_0^{PT}|$, where $\rho_0$ is minimally degenerate. In both
instances, once again applying the density-approximation procedure of Provost 
\cite{Provost}, we will find separability probabilities equal to {\it one-half} those obtained by use of the concise formula for $P(\alpha)$ ((\ref{Hou1})-(\ref{Hou3})). We, then, show the consistency of these results with a theorem of Szarek, Bengtsson and {\.Z}yczkowski
\cite{sbz}, thus, lending even further support to that already compiled for the validity
of the formula for $P(\alpha)$.
\section{Generic, nondegenerate cases ($|\rho| \neq 0$)} \label{Nondegenerate}
In the course of obtaining the $5F4$-hypergeometric-based Hilbert-Schmidt (HS) moment formula above--and a more general two-variable ($n, k$) form of it for $\left\langle \left\vert \rho^{PT}\right\vert
^{n}\left\vert \rho\right\vert ^{k}\right\rangle /\left\langle \left\vert
\rho\right\vert ^{k}\right\rangle$--there were employed certain intermediate ``utility functions'', in particular \cite[p. 26]{MomentBased}, to use the notation there,
\begin{align*}
F_{2}\left(  n,k\right)   &  =\left\langle \left\vert \rho\right\vert
^{k}\left(  \left\vert \rho^{PT}\right\vert -\left\vert \rho\right\vert
\right)  ^{n}\right\rangle /\left\langle \left\vert \rho\right\vert
^{k}\right\rangle, 
\end{align*}
incorporating the new variable of specific interest here, that is, 
$(|\rho^{PT}| -|\rho|)$.
Subsequently, we have obtained  the explicit formula (Appendix A)
\begin{align*}
F_{2}\left(  n,k\right)   &  =
\left(  -1\right)  ^{n}\frac{\left(  \alpha\right)  _{n}\left(  \alpha
+\frac{1}{2}\right)  _{n}\left(  n+2k+2+5\alpha\right)  _{n}}{2^{4n}\left(
k+3\alpha+\frac{3}{2}\right)  _{n}\left(  2k+6\alpha+\frac{5}{2}\right)
_{2n}}\\
& \times~_{4}F_{3}\left(
\genfrac{}{}{0pt}{}{-\frac{n}{2},\frac{1-n}{2},k+1+\alpha,k+1+2\alpha
}{1-n-\alpha,\frac{1}{2}-n-\alpha,n+2k+2+5\alpha}%
;1\right)  .
\end{align*}

We set $k=0$ in this formula, and once again applied the Legendre-polynomial-based-density-approximation procedure of Provost \cite{Provost}, in the same manner as in our previous studies \cite{MomentBased,slaterJModPhys}. It was first necessary to observe, however, that rather than the 
variable 
range $-\frac{1}{16} \leq |\rho^{PT}| \leq \frac{1}{256}$ employed in these earlier studies, the appropriate interval would now be \newline $-\frac{1}{16} \leq (|\rho^{PT}|-|\rho|) \leq \frac{1}{432}$. (Note that $432 = 2^4 \cdot 3^3$, as well as, of course, $16=2^4$ and $256 =2^8$.).
\subsection{Two-parameter family of density matrices illustrating range of 
$(|\rho^{PT}|-|\rho|)$}
The extreme values of this indicated range $[-\frac{1}{16}, \frac{1}{432}]$ can be illustrated by the use of a two-parameter family of density matrices
\begin{equation} \label{example2}
\rho = \left(
\begin{array}{cccc}
 u & 0 & 0 & v \\
 0 & \frac{1}{2}-u & 0 & 0 \\
 0 & 0 & \frac{1}{2}-u & 0 \\
 v & 0 & 0 & u \\
\end{array}
\right). 
\end{equation}
For $u =\frac{1}{4}$, $v=0$, we have the limiting value, $|\rho|=\frac{1}{256}$, while for $u = \frac{1}{2}$, $v=\frac{1}{2}$, we have the limit $|\rho^{PT}|=-\frac{1}{16}$. Further, for $u=\frac{1}{6}$, 
$v=  \frac{1}{6}$, we have both $|\rho|=0$ and $|\rho^{PT}| =\frac{1}{432}$.
If for this last choice of parameters, we interchange 
$|\rho|$ with its partial transpose $|\rho^{PT}|$, a value of $-\frac{1}{432}$, that is, the lower bound on the domain of separability, is obtained for the variable 
$(|\rho^{PT}|-|\rho|)$ of current interest. 

As examples of entangled states for which the values of $\left(  \left\vert \rho
^{PT}\right\vert -\left\vert \rho\right\vert \right)  $ are dense in $\left[
-\frac{1}{16},0\right]  $, we can employ  the above family (\ref{example2}) 
with $u=\frac{1}%
{2}-\varepsilon,2\varepsilon\leq v\leq\frac{1}{2}-2\varepsilon$, where
$0<\varepsilon<\frac{1}{6}$ so that $\frac{1}{2}-u<v<u$. (Note that $\rho$ is
positive-definite provided $0<u<\frac{1}{2}$ and $\left\vert v\right\vert
<u$.) For this family, $\left\vert \rho^{PT}\right\vert =\left(  \frac{1}%
{2}-\varepsilon\right)  ^{2}\left(  \varepsilon^{2}-v^{2}\right)  <0$ and
$\left\vert \rho^{PT}\right\vert -\left\vert \rho\right\vert =-v^{2}\left(
4u-1\right)  $. Thus, the range for the given parameters is $\left[  -\frac
{1}{16}\left(  1-4\varepsilon\right)  ^{3},-\varepsilon^{2}\left(
1-4\varepsilon\right)  \right]  $. Let $\varepsilon\rightarrow0_{+}$ to get
the interval of entanglement $\left[  -\frac{1}{16},0\right]  $.

We crucially rely throughout these series of analyses upon the  proposition that 
$|\rho^{PT}|>0$ is both a necessary and sufficient condition for a two-qubit state to be separable \cite{augusiak,Demianowicz}. To expand upon this point, the partial transpose of a $4 \times 4$ density matrix $\rho$ can possess at most one negative eigenvalue, so that the 
non-negativity of $|\rho^{PT}|$--the product of the four eigenvalues of $\rho^{PT}$--is 
tantamount to separability. 
\subsubsection{Intervals of interest in $6 \times 6$ density matrix case}
Quite contrastingly, and more complicatedly, in our ongoing study of generic 
(generalized qubit-{\it qutrit}) $6 \times 6$ density matrices endowed with the Hilbert-Schmidt measure \cite{Hybrid},  $|\rho^{PT}|$ can be {\it either} positive or negative 
for an entangled state. This is due to the possibility that {\it two} eigenvalues of $\rho^{PT}$ could now be negative. In this $6 \times 6$ case, it appears that the ranges of interest are $0 \leq |\rho| \leq (\frac{1}{6})^6 = \frac{1}{46656}$ and $-\frac{1}{2916} \leq |\rho^{PT}| \leq \frac{1}{2916}$, where $2916 = 2^2 \times 3^6$.
The interval of entanglement $-\frac{1}{2916} \leq |\rho^{PT}| \leq 0$ would be associated with a single negative eigenvalue, and $(\frac{1}{6})^6 \leq |\rho^{PT}| \leq \frac{1}{2916}$ with a pair of negative eigenvalues. The remaining segment $0 \leq |\rho^{PT}| \leq (\frac{1}{6})^6$ could have partial transposes having none--indicating 
separability--or two negative eigenvalues.
\subsection{Separability probability calculations, using density-approximation} \label{firstsepprobs}
In the generalized $4 \times 4$ density matrix scenario, the variable  $(|\rho^{PT}|-|\rho|)$ ranges over
$[-\frac{1}{16},\frac{1}{432}]$, with the the subrange $[0, \frac{1}{432}]$ of $(|\rho^{PT}|-|\rho|)$
containing only separable states. Now, employing 
$\alpha =1$ in the new $4F3$ hypergeometric-based  moment formula 
immediately above, we obtained,
based on 9,451 ($n =0, 1,\ldots 9,450$) moments, again using the Provost 
density-approximation methodology \cite{Provost}, an estimate for the separability probability (over $[-\frac{1}{16},\frac{1}{432}]$) that was 0.50000004358 as large as $P(1) = \frac{8}{33}$, given by eqs. (\ref{Hou1})-(\ref{Hou3}).
The parallel calculations in the two-rebit ($\alpha = \frac{1}{2}$) and two-``quaterbit''
($\alpha = 2$) cases yielded estimates of $0.5000025687 \times P(\frac{1}{2})$ 
and $0.5000000000177 \times P(2)$, respectively. (Differences in rates of convergence--much the same as observed in \cite{MomentBased}--can be attributed to the initial [zero{\it th}-order] assumption of the Legendre-polynomial-density-approximation procedure that the probability distributions to be fitted are uniform in nature, rendering more sharply-peaked distributions more difficult to rapidly approximate well.) {\it A fortiori}, for the $\alpha =4$ (conjecturally {\it octonionic}) value \cite[p. 9]{slaterJModPhys}, $P(4)= \frac{4482}{4091349}$, our computed value  was  $0.500000000000000015 \times P(4).$ These outcomes, certainly, help to strongly bolster the validity of the (yet formally unproven) concise formula ((\ref{Hou1})-(\ref{Hou3})), yielding the full (whole) generic Hilbert-Schmidt two-qubit separability probabilities $P(\alpha)$.

In Fig.~\ref{fig:BisectionPlot} we display an estimate based on the first 51 $(n =0,\ldots,50)$ moments of the probability distributions under analysis as a function of $\alpha$ over the subrange $[-\frac{1}{108},\frac{1}{432}]$ of the full range $[-\frac{1}{16},\frac{1}{432}]$ of 
$(|\rho^{PT}|-|\rho|)$. The distributions are more sharply peaked for smaller $\alpha$ (nearer to $\alpha = \frac{1}{2}$ in the plot), as the larger values of 
$P(\alpha)$ for smaller $\alpha$ would indicate (Appendix E).
\begin{figure}
\includegraphics{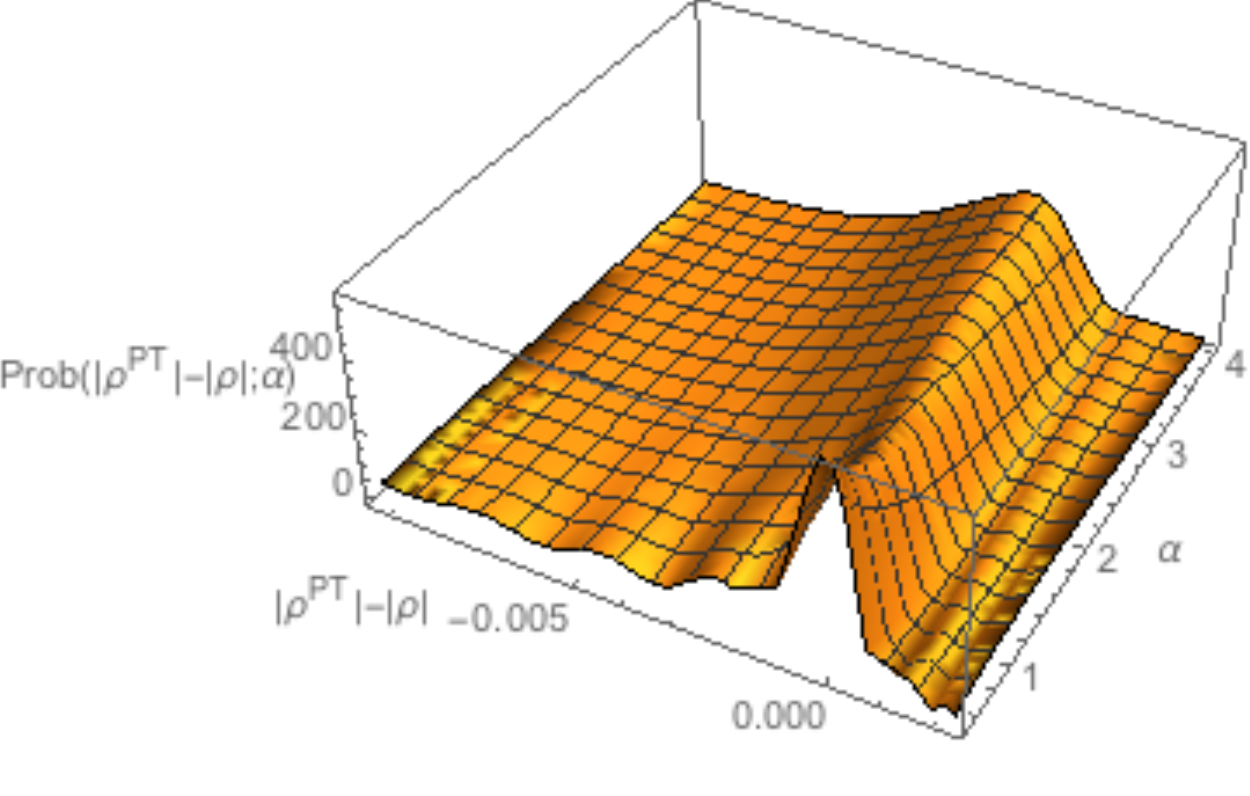}
\caption{\label{fig:BisectionPlot}Density-approximation estimation based on the first 51 moments of the probability distributions, as a function of the Dyson-index-like parameter $\alpha$, of the 
variable ($|\rho^{PT}|-|\rho|$)}
\end{figure}
\subsubsection{Probabilities over larger interval $[-\frac{1}{432}, \frac{1}{432}]$}
For the two-rebit, two-qubit and two-quaterbit probabilities over the extended interval
$[-\frac{1}{432}, \frac{1}{432}]$, symmetric about zero, containing all separable and now some entangled states (and thus providing {\it upper} bounds on the total separability probabilities), the estimates, again based on 9,451 moments were 0.78082617689, 0.69244685258 and 0.601390039979. However, we were not able to discern any particular underlying common structure (formula) in these values. 
\section{Generic, minimally degenerate cases ($|\rho| = 0$)} \label{Degenerate}
Let us now importantly note that these ``half-separability-probabilities'' of $\frac{29}{128}, \frac{4}{33}, \frac{13}{323}$ obtained above, appear, by Hilbert-Schmidt-based analyses of Szarek, Bengtsson and {\.Z}yczkowski \cite{sbz}, to be {\it exactly} equal to the ``full-separability-probabilities''
for the corresponding minimally-degenerate (boundary, that is $|\rho|=0$) generic two-rebit, two-qubit and 
two-quaterbit states. We are now able to make a further interesting connection to this body of work--and thereby find additional strong support for its findings, as well as our earlier ones \cite{MomentBased,slaterJModPhys},  obtained quite independently.

Let us note, firstly, that in \cite[sec. 7]{MomentBased} it was asserted that the range of $|\rho^{PT}|$ under the minimally-degenerate constraint $|\rho|=0$ is (once again, as it was for ($|\rho^{PT}|-|\rho|)$ above) the interval 
$[-\frac{1}{16},\frac{1}{432}]$. (Under this determinantal constraint, we will employ the notation 
$|\rho_0^{PT}|$.) 
The maximum of this range is attainable by the two-parameter density matrix (\ref{example2}), for example, with $u =\frac{1}{6}$ and $v= \frac{1}{6}$. 

In \cite[App. C]{MomentBased}, we had listed the two-rebit ($\alpha= \frac{1}{2}$) Hilbert-Schmidt moments of $|\rho_0^{PT}|^n$, $n=1,\ldots,10$. Now, in an exploratory exercise, we computed the ratio of these ten results to the corresponding moments given by the $4F3$-based  formula above for the moments of $(|\rho^{PT}|-|\rho|)$. Most interestingly, these ten ratios had the
explanatory formula (found by the Mathematica FindSequenceFunction command)
\begin{equation} \label{simpleratio}
\frac{(3 n+7) (4 n+9)}{9 (4 n+7)}.
\end{equation}
Then, performing further computations for $\alpha =1, n =1, 2$, it was possible to develop a line of reasoning (Appendix B)
that  the expression 
(\ref{simpleratio}) was the $\alpha = \frac{1}{2}$-specific case of a more 
general moment formula, incorporating the factor 
\begin{equation} \label{fullygeneral}
\frac{(10 \alpha +3 n+2) (12 \alpha +4 n+3)}{(12 \alpha +3) (10 \alpha +4 n+2)} 
\end{equation}
(equaling 1 for $n=0$). 
In fact, the existence of a ratio of this form between the moments implies the equality of the probabilities that the respective random variables--in the case at hand, ($|\rho^{PT}| -|\rho|$) and $|\rho_0^{PT}|$--are positive (Appendix C).
\subsection{Separability probability calculations, using density-approximation}
We employed 9,451 of the original $4F3$-based moments now adjusted by this last ratio 
(\ref{fullygeneral}), in the density-approximation routine of Provost, just as before. For $\alpha =\frac{1}{2}, 1, \frac{3}{2}$ and 2, we obtained for the cumulative probabilities over the separability interval $[0,\frac{1}{432}]$, the values $0.50000261669 \times P(\frac{1}{2})$, $0.50000003530 \times P(1)$, $0.50000000060467 \times P(\frac{3}{2})$ and 
$0.50000000001267 \times P(2)$, and 
similarly for $\alpha > 2$. 

We note that the convergence of these results to 
$P(\alpha)/2$ is somewhat superior than in the earlier parallel set of analyses 
for ($|\rho^{PT}|-|\rho|$) (sec.~\ref{firstsepprobs}). Apparently relatedly, the ratio of the standard deviation of the probability distribution of ($|\rho^{PT}|-|\rho|$) to that of $|\rho_0^{PT}|$ is 0.788 for $\alpha = \frac{1}{2}$ and 0.857 for $\alpha =1$. So, the distribution for 
($|\rho^{PT}|-|\rho|$) is more peaked at the value zero. Thus, the Legendre-polynomial-based density-approximation procedure (which starts with a uniform approximation) is slower to converge in those cases. Further consistent with this observation, based now on 6,301 moments, the ratio  of the ``$y$''-intercept for  ($|\rho^{PT}|-|\rho|$) to that for $|\rho_0^{PT}|$ was estimated as 1.30202 for $\alpha =\frac{1}{2}$ and 1.21134 for $\alpha =1$.
\subsubsection{Relations to separability-probability theorem of Szarek, Bengtsson and {\.Z}yczkowski}
These density approximation estimations  extraordinarily close to one-half 
certainly are strongly in line with
the main finding  regarding the Hilbert-Schmidt separability probabilities of minimally degenerate (boundary) 
states of Szarek, Bengtsson and {\.Z}yczkowski \cite[Theorem 2]{sbz}.
These three authors had established that the set of positive-partial-transpose states for an arbitrary bipartite systems is ``pyramid-decomposable'' and hence, a body of ``constant height''. They stated that ``since our reasoning hinges directly on the Euclidean geometry, it does not allow one to predict any values of analogous ratios computed with respect to the Bures measure, nor other measures'' \cite[p. L125]{sbz}.

Nonetheless, the ``symmetric halves'' separability-probability finding elucidated above 
(that is, the separability probability for $|\rho^{PT}| >|\rho|$ equaling that for 
$|\rho|>|\rho^{PT}|$) does appear to be measure-independent, that is extendible from the Hilbert-Schmidt  (flat, Euclidean) metric to the use of alternative metrics, such as the Bures (minimal monotone) metric \cite{szBures,ingemarkarol}. 
\subsubsection{Rank-two $4 \times 4$ density matrices}
We have also been able to conclude that for the generic 
rank-{\it two} $4 \times 4$ density matrices (for which, of course, $|\rho|$ is also zero)--as opposed to the generic rank-three 
(minimally degenerate) ones just analyzed,
the Hilbert-Schmidt separability probability is zero. An intuitive argument to this effect is that if one possesses a rank-two $4 \times 4$ density matrix with a positive partial transpose, then if one interchanges the role of these two matrices, one has a partial transpose with two zero eigenvalues (cf. \cite{Ishizaka}). Such a scenario is infinitesimally close to one with two (slightly) negative eigenvalues--a situation which has been well-established is not tenable \cite{augusiak,Demianowicz}. (Somewhat contrastingly, in \cite{SlaterPRA}, numerical evidence indicated that the ratio of Hilbert-Schmidt separability probabilities for generic [rank-six] $6 \times 6$ density matrices to rank-four such matrices was close to the integer 34.) {\it A fortiori}, the Hilbert-Schmidt separability probability of the generic rank-one (pure states) $4 \times 4$ density matrices is also zero.
\section{Concluding Remarks}
In need of further study is the issue of whether or not the Dyson-index {\it ansatz} of random matrix theory \cite{dumitriu,caselle}--apparently applicable in the Hilbert-Schmidt case, as our various results for general $\alpha$ so far would indicate--extends to other measures (Bures,\ldots), as well (cf. \cite{BuresHilbert,Hybrid}).

We note, regretfully, of course, that formal proofs for the Hilbert-Schmidt determinantal moment formulas and the density-approximation results obtained with their use have not yet been advanced--and certainly still seem far from development. Certainly, however, the cumulative computational evidence appears very strong for the validity of, {\it inter alia},  
the indicated $\frac{29}{64}, \frac{8}{33}$ and $\frac{26}{323}$  two-rebit, two-qubit and two-quaterbit separability probabilities. Noticeably still lacking 
is an insightful geometric intuition into the geometry
of the $4 \times 4$ density matrices that might help to explain such results (cf. \cite{avron,avron2,holik,milne,sarbicki}). Can the $\frac{8}{33}$ two-qubit separability probability result, for example, only be understood in some sense as a limiting result--as the infinite-summation ``concise formula'' ((\ref{Hou1})-(\ref{Hou3})) might seem to indicate--or is it possibly remarkably manifest in some {\it discrete} (pyramidal? \cite{sbz}) subdivision of the 15-dimensional convex set of $4 \times 4$ density matrices?

Possible extensions of the research program presented above and in \cite{MomentBased,slaterJModPhys} to the Hilbert-Schmidt case of $6 \times 6$ (qubit-qutrit) density matrices and the Bures instance of $4 \times 4$ density matrices have been investigated in 
\cite{BuresHilbert,Hybrid}. Some limited determinantal moment computations have been reported 
($\alpha = \frac{1}{2}$, $n=1, 2$; $\alpha =1$, $n=1$, in both instances) (Appendix D), but yet no comparable formulas of the type 
$\left\langle \left\vert \rho\right\vert ^{k} \left\vert \rho^{PT}\right\vert
^{n}\right\rangle /\left\langle \left\vert
\rho\right\vert ^{k}\right\rangle$ nor $\left\langle \left\vert \rho\right\vert ^{k} (\left\vert \rho^{PT}\right\vert -|\rho|)^{n}\right\rangle /\left\langle \left\vert
\rho\right\vert ^{k}\right\rangle$ developed. Such formulas have been the fundamental basis for most of the advances noted here and previously \cite{MomentBased,slaterJModPhys}.\\
\bigskip
\noindent {\bf{Appendix A: Moments of $\left(  \left\vert \rho^{PT}\right\vert -\left\vert
\rho\right\vert \right)  $}}

Consider the general $\alpha$ case, generic $k$.

Let%
\[
g\left(  k,n\right)  :=\frac{\left(  k+1\right)  _{n}\left(  k+1+\alpha
\right)  _{n}\left(  k+1+2\alpha\right)  _{n}}{2^{6n}\left(  k+3\alpha
+\frac{3}{2}\right)  _{n}\left(  2k+6\alpha+\frac{5}{2}\right)  _{2n}},
\]
there is a multiplication relation:%
\[
g\left(  0,k\right)  g\left(  k,n\right)  =g\left(  0,k+n\right)  .
\]
Let%
\[
h\left(  k,n\right)  :=~_{5}F_{4}\left(
\genfrac{}{}{0pt}{}{-n,-k,\alpha,\alpha+\frac{1}{2},-2k-2n-1-5\alpha
}{-k-n-\alpha,-k-n-2\alpha,-\frac{k+n}{2},-\frac{k+n-1}{2}}%
;1\right)  .
\]
Then%
\begin{align*}
\left\langle \left\vert \rho\right\vert ^{k}\right\rangle  &  =g\left(
0,k\right) \\
\left\langle \left\vert \rho^{PT}\right\vert ^{n}\left\vert \rho\right\vert
^{k}\right\rangle /\left\langle \left\vert \rho\right\vert ^{k}\right\rangle
&  =g\left(  k,n\right)  h\left(  k,n\right) \\
\left\langle \left\vert \rho^{PT}\right\vert ^{n}\left\vert \rho\right\vert
^{k}\right\rangle  &  =g\left(  0,k+n\right)  h\left(  k,n\right)  .
\end{align*}
Define%
\[
F_{2}\left(  n,k\right)  =\left\langle \left\vert \rho\right\vert ^{k}\left(
\left\vert \rho^{PT}\right\vert -\left\vert \rho\right\vert \right)
^{n}\right\rangle /\left\langle \left\vert \rho\right\vert ^{k}\right\rangle
,
\]
then%
\begin{align*}
F_{2}\left(  n,k\right)   &  =\frac{1}{g\left(  0,k\right)  }\sum_{j=0}%
^{n}\binom{n}{j}\left(  -1\right)  ^{n-j}\left\langle \left\vert
\rho\right\vert ^{k+n-j}\left\vert \rho^{PT}\right\vert ^{j}\right\rangle \\
&  =\frac{1}{g\left(  0,k\right)  }\sum_{j=0}^{n}\binom{n}{j}\left(
-1\right)  ^{n-j}g\left(  0,k+n\right)  h\left(  k+n-j,j\right) \\
&  =g\left(  k,n\right)  \sum_{j=0}^{n}\binom{n}{j}\left(  -1\right)
^{n-j}h\left(  k+n-j,j\right)  .
\end{align*}
We will produce $F_{2}^{\prime}\left(  n,k\right)  :=\sum_{j=0}^{n}\binom
{n}{j}\left(  -1\right)  ^{n-j}h\left(  k+n-j,j\right)  $ as a single sum (so
that $F_{2}\left(  n,k\right)  =g\left(  k,n\right)  F_{2}^{\prime}\left(
n,k\right)  $).

\begin{lemma}
Let $n,m=0,1,2,\ldots$ and let $x$ be a variable, if $0\leq m\leq n$ then%
\[
\sum_{j=0}^{n}\frac{\left(  -n\right)  _{j}}{j!}\left(  -j\right)  _{m}\left(
x+j\right)  _{m}=\left(  -1\right)  ^{m}\frac{\left(  x\right)  _{2m}}{\left(
x\right)  _{n}}\left(  -n\right)  _{m}\left(  -m\right)  _{n-m},
\]
otherwise the sum is zero.
\end{lemma}

\begin{proof}
If $m>n$ then $\left(  -j\right)  _{m}=0$ for $0\leq j\leq n$. Suppose $0\leq
m\leq n$ then $\left(  -j\right)  _{m}=0$ for $0\leq j<m$ and the sum is over
$m\leq j\leq n$. Thus
\begin{align*}
\sum_{j=m}^{n}\frac{\left(  -n\right)  _{j}}{j!}\left(  -j\right)  _{m}\left(
x+j\right)  _{m}  &  =\left(  -1\right)  ^{m}\sum_{j=m}^{n}\frac{\left(
-n\right)  _{j}~j!}{j!\left(  j-m\right)  !}\frac{\left(  x\right)
_{j}\left(  x+j\right)  _{m}}{\left(  x\right)  _{j}}\\
&  =\left(  -1\right)  ^{m}\left(  x\right)  _{m}\sum_{j=m}^{n}\frac{\left(
-n\right)  _{j}\left(  x+m\right)  _{j}}{\left(  j-m\right)  !\left(
x\right)  _{j}}.
\end{align*}
Change the index of summation $j=m+i$ then the sum equals%
\begin{align*}
&  \left(  -1\right)  ^{m}\frac{\left(  -n\right)  _{m}\left(  x\right)
_{m}\left(  x+m\right)  _{m}}{\left(  x\right)  _{m}}\sum_{i=0}^{n-m}%
\frac{\left(  m-n\right)  _{i}\left(  x+2m\right)  _{i}}{i!\left(  x+m\right)
_{i}}\\
&  =\left(  -1\right)  ^{m}\frac{\left(  -n\right)  _{m}\left(  x\right)
_{2m}}{\left(  x\right)  _{m}}\frac{\left(  -m\right)  _{n-m}}{\left(
x+m\right)  _{n-m}}\\
&  =\left(  -1\right)  ^{m}\left(  -n\right)  _{m}\left(  -m\right)
_{n-m}\frac{\left(  x\right)  _{2m}}{\left(  x\right)  _{n}},
\end{align*}
by the Chu-Vandermonde sum.
\end{proof}

Observe that $\left(  -m\right)  _{n-m}=0$ for $2m<n$. Then%
\begin{align*}
F_{2}^{\prime}\left(  n,k\right)   &  =\left(  -1\right)  ^{n}\sum_{j=0}%
^{n}\frac{\left(  -n\right)  _{j}}{j!}\\
&  \times\sum_{i=0}^{n}\frac{\left(  -j\right)  _{i}\left(  j-k-n\right)
_{i}\left(  \alpha\right)  _{i}\left(  \alpha+\frac{1}{2}\right)  _{i}\left(
-2k-2n-1-5\alpha\right)  _{i}}{i!\left(  -k-n-\alpha\right)  _{i}\left(
-k-n-2\alpha\right)  _{i}\left(  -\frac{k+n}{2}\right)  _{i}\left(
-\frac{k+n-1}{2}\right)  _{i}}\\
&  =\left(  -1\right)  ^{n}\sum_{i=0}^{n}\frac{\left(  \alpha\right)
_{i}\left(  \alpha+\frac{1}{2}\right)  _{i}\left(  -2k-2n-1-5\alpha\right)
_{i}}{i!\left(  -k-n-\alpha\right)  _{i}\left(  -k-n-2\alpha\right)
_{i}\left(  -\frac{k+n}{2}\right)  _{i}\left(  -\frac{k+n-1}{2}\right)  _{i}%
}\\
&  \times\sum_{j=0}^{n}\frac{\left(  -n\right)  _{j}}{j!}\left(  -j\right)
_{i}\left(  j-k-n\right)  _{i}.
\end{align*}
Apply the lemma to the $j$-sum with $x=-k-n$ and $m=i$ to obtain%
\[
\left(  -1\right)  ^{i}\left(  -n\right)  _{i}\left(  -i\right)  _{n-i}%
\frac{\left(  -n-k\right)  _{2i}}{\left(  -n-k\right)  _{n}}=\left(
-1\right)  ^{i}\frac{\left(  -n\right)  _{i}\left(  -i\right)  _{n-i}}{\left(
-n-k\right)  _{n}}2^{2i}\left(  -\frac{k+n}{2}\right)  _{i}\left(
-\frac{k+n-1}{2}\right)  _{i}%
\]
and thus
\[
F_{2}^{\prime}\left(  n,k\right)  =\frac{\left(  -1\right)  ^{n}}{\left(
-n-k\right)  _{n}}\sum_{i=0}^{n}\frac{\left(  -n\right)  _{i}\left(
-i\right)  _{n-i}\left(  \alpha\right)  _{i}\left(  \alpha+\frac{1}{2}\right)
_{i}\left(  -2k-2n-1-5\alpha\right)  _{i}}{i!\left(  -k-n-\alpha\right)
_{i}\left(  -k-n-2\alpha\right)  _{i}}\left(  -1\right)  ^{i}2^{2i}.
\]
This is not in hypergeometric form because of the term $\left(  -i\right)
_{n-i}$; also the summation extends over $\frac{n}{2}\leq i\leq n$. Change the
index $j=n-i$ then
\begin{align*}
\frac{\left(  -n\right)  _{i}}{i!}\left(  -i\right)  _{n-i} &  =\left(
-1\right)  ^{i}\frac{n!}{\left(  n-i\right)  !i!}\left(  -1\right)
^{n-i}\frac{i!}{\left(  2i-n\right)  !}=\left(  -1\right)  ^{n}\frac
{n!}{j!\left(  n-2j\right)  !}\\
&  =\left(  -1\right)  ^{n}\frac{2^{2j}}{j!}\left(  -\frac{n}{2}\right)
_{j}\left(  \frac{1-n}{2}\right)  _{j}%
\end{align*}
and the reversal formula is%
\begin{align*}
\left(  x\right)  _{i} &  =\left(  x\right)  _{n-j}=\frac{\left(  x\right)
_{n-j}\left(  x+n-j\right)  _{j}}{\left(  x+n-j\right)  _{j}}\\
&  =\left(  -1\right)  ^{j}\frac{\left(  x\right)  _{n}}{\left(  1-n-x\right)
_{j}}.
\end{align*}
Thus%
\begin{align}
F_{2}^{\prime}\left(  n,k\right)   &  =\frac{\left(  -1\right)  ^{n}\left(
\alpha\right)  _{n}\left(  \alpha+\frac{1}{2}\right)  _{n}\left(
-2k-2n-1-5\alpha\right)  _{n}}{\left(  -n-k\right)  _{n}\left(
-k-n-\alpha\right)  _{n}\left(  -k-n-2\alpha\right)  _{n}}\label{F2p}\\
&  \times\sum_{j=0}^{[n/2]}\frac{\left(  -\frac{n}{2}\right)  _{j}\left(  \frac
{1-n}{2}\right)  _{j}\left(  k+1+\alpha\right)  _{j}\left(  k+1+2\alpha
\right)  _{j}}{j!\left(  1-n-\alpha\right)  _{j}\left(  \frac{1}{2}%
-n-\alpha\right)  _{j}\left(  n+2k+2+5\alpha\right)  _{j}}2^{2j+2n-2j}%
\nonumber\\
&  =\left(  -1\right)  ^{n}2^{2n}\frac{\left(  \alpha\right)  _{n}\left(
\alpha+\frac{1}{2}\right)  _{n}\left(  n+2k+2+5\alpha\right)  _{n}}{\left(
k+1\right)  _{n}\left(  k+1+\alpha\right)  _{n}\left(  k+1+2\alpha\right)
_{n}}\nonumber\\
&  \times~_{4}F_{3}\left(
\genfrac{}{}{0pt}{}{-\frac{n}{2},\frac{1-n}{2},k+1+\alpha,k+1+2\alpha
}{1-n-\alpha,\frac{1}{2}-n-\alpha,n+2k+2+5\alpha}%
;1\right)  ;\nonumber
\end{align}
a balanced sum.

The formula was tested for $F_{2}\left(  2,k\right)  $, also directly verified
for $n=3$, arbitrary $\alpha$.

Combining the front factors in $F_{2}\left(  n,k\right)  $ (from $g\left(
k,n\right)  $) we obtain%
\begin{equation}
\left(  -1\right)  ^{n}\frac{\left(  \alpha\right)  _{n}\left(  \alpha
+\frac{1}{2}\right)  _{n}\left(  n+2k+2+5\alpha\right)  _{n}}{2^{4n}\left(
k+3\alpha+\frac{3}{2}\right)  _{n}\left(  2k+6\alpha+\frac{5}{2}\right)
_{2n}}.\label{F2fact}%
\end{equation}
\bigskip
{\bf{Appendix B: Minimally degenerate case}}

Recall the Cholesky decomposition $\rho=C^{\ast}C$ where%
\begin{equation}
C=%
\begin{bmatrix}
x_{1} & x_{5} & x_{6} & x_{7}\\
0 & x_{2} & x_{8} & x_{9}\\
0 & 0 & x_{3} & x_{10}\\
0 & 0 & 0 & x_{4}%
\end{bmatrix}
\label{Chol}%
\end{equation}
with $\sum_{j=1}^{10}\left\vert x_{j}\right\vert ^{2}=1,x_{j}\geq0$ for $1\leq
j\leq4$ and $x_{j}\in\mathbb{R}$ for $\alpha=\frac{1}{2}$, or $x_{j}%
\in\mathbb{C}$ for $\alpha=1$ ($5\leq j\leq10$). Then $\left\vert
\rho\right\vert =\left(  x_{1}x_{2}x_{3}x_{4}\right)  ^{2}$. Denote
integration over the space of $4\times4$ positive-definite matrices with trace
one by $\left\langle \cdot\right\rangle $. Suppose $p$ is a monomial in
$\left\{  x_{j},\overline{x_{j}}\right\}  $ then%
\begin{equation}
\left\langle p\right\rangle =\frac{\left(  1+3\alpha\right)  _{n_{1}}\left(
1+2\alpha\right)  _{n_{2}}\left(  1+\alpha\right)  _{n_{3}}\left(  1\right)
_{n_{4}}}{\left(  4+12\alpha\right)  _{n}}\prod_{j=5}^{10}\left(
\alpha\right)  _{n_{j}}\label{monom}%
\end{equation}
where $n=\sum_{j=1}^{10}n_{j}$ and:

\begin{itemize}
\item $\alpha=\frac{1}{2},~p=\prod_{j=1}^{10}x_{j}^{2n_{j}}$, that is, each
exponent is even;

\item $\alpha=1,~p=\prod_{j=1}^{4}x_{j}^{2n_{j}}\times\prod_{j=5}^{10}\left(
x_{j}\overline{x_{j}}\right)  ^{n_{j}}$, that is, $p$ is a monomial in
$x_{1}^{2},\cdots,x_{4}^{2},\left\vert x_{5}\right\vert ^{2},\cdots,\left\vert
x_{10}\right\vert ^{2}$;
\end{itemize}

otherwise $\left\langle p\right\rangle =0$.

The boundary of the set of states (positive-definite matrices with trace one)
contains $\Omega_{0}$, the set of positive-semidefinite matrices of possible
ranks $1,2$ or $3$ (and trace one). In the following discussion the parameters
are stated first for the real $\alpha=\frac{1}{2}$ case, then in parentheses
for the complex $\alpha=1$ case. We consider the determinant of the partial
transpose, denoted by $\left\vert \rho_{0}^{PT}\right\vert $, as a random
variable defined on $\Omega_{0}$, with respect to the Hilbert-Schmidt measure,
that is, the Euclidean $8$-dimensional (resp. $14$) measure, restricted to
$\Omega_{0}$. We claim that integrating with respect to this measure can be
carried out by integrating over Cholesky products with $x_{4}=0$ and the
surface measure on the unit sphere in $\mathbb{R}^{9}$ (resp. $\mathbb{R}%
^{15}$) and the Jacobian $x_{1}^{5/2}x_{2}^{2}x_{3}^{3/2}$ (resp. $x_{1}%
^{4}x_{2}^{3}x_{3}^{2}$). The generic (or random) elements of $\Omega_{0}$ are
called minimally degenerate.

Let $\Omega_{0}^{C}$ denote the set of Cholesky products $C^{\ast}C$ with the
conditions as in (\ref{Chol}) and with $x_{4}=0$. The same arguments used in
\cite[sec. D.2]{MomentBased} show that the surface measure on the sphere multiplied by the
above Jacobian is mapped to the HS-measure on $\Omega_{0}^{C}$. So it remains
to show that the elements of $\Omega_{0}\backslash\Omega_{0}^{C}$ do not enter
into the probability calculation. For any real symmetric (resp. Hermitian)
$4\times4$ matrix $M$ let $d_{j}\left(  M\right)  $ denote the determinant of
the upper left $j\times j$ submatrix of $M$ (a principal minor), for $1\leq
j\leq4$. Then $M$ is positive-definite (resp. positive-semidefinite) if and
only if $d_{j}\left(  M\right)  >0$ for all $j$ (resp. $d_{j}\left(  M\right)
\geq0$ for all $j$). Suppose $\rho\in\Omega_{0}$ then $\rho$ has a unique
Cholesky factorization if $d_{j}\left(  \rho\right)  >0$ for $j=1,2,3$; these
conditions imply that $x_{1}x_{2}x_{3}>0$. As a consequence $x_{4}=0$ because
$\left\vert \rho\right\vert =\left(  x_{1}x_{2}x_{3}x_{4}\right)  ^{2}$, and
thus $\rho\in\Omega_{0}^{C}$.

As contrapositive we have shown that $\rho\in\Omega_{0}\backslash\Omega
_{0}^{C}$ implies $d_{j}\left(  \rho\right)  =0$ for at least one value of
$j=1,2,3$. This is an additional algebraic condition besides $\left\vert
\rho\right\vert =0$ satisfied by the entries of $\rho$, that is $\rho$ belongs
to a manifold (or variety) of lower dimension in $\Omega_{0}$ ($<8$ for
$\alpha=\frac{1}{2}$, and $<14$ for $\alpha=1$) and such sets have HS-measure zero.

The appropriate measure on the set of Cholesky factors with $x_{4}=0$ can be
interpreted as a conditional density on a subset of the unit sphere, or as the
surface measure on the sphere in one less dimension. In general suppose $f$ is
a density on some region $E\subset\mathbb{R}^{N}$ then the conditional density
given $y_{N}=u$ is%
\[
\frac{f\left(  y_{1},\cdots,y_{N-1},u\right)  }{\int_{E_{u}}f\left(
t_{1},\cdots,t_{N-1},u\right)  dt_{1}\cdots dt_{N-1}}%
\]
where $E_{u}:=\left\{  t\in\mathbb{R}^{N-1}:\left(  t,u\right)  \in E\right\}
$. In our situation the density vanishes on the set of interest ($x_{4}=0$) so
we need to take a limit.

Consider a general Dirichlet density%
\[
\frac{\Gamma\left(  \alpha_{1}+\cdots+\alpha_{N}\right)  }{\Gamma\left(
\alpha_{1}\right)  \cdots\Gamma\left(  \alpha_{N}\right)  }\prod_{j=1}%
^{N}y_{j}^{\alpha_{j}-1}%
\]
on $T_{N-1}=\left\{  y\in\mathbb{R}_{N-1}:y_{j}\geq0~\forall j,\sum
_{j=1}^{N-1}y_{j}\leq1\right\}  $ and $y_{N}:=1-\sum_{j=1}^{N-1}y_{j}$; also
$\alpha_{j}>0~\forall j$. Compute the conditional density given $x_{N}=u$ with
$0<u<1$; then $E_{u}=\left(  1-u\right)  T_{N-2}$, that is, $E_{u}=\left\{
y\in\mathbb{R}_{N-2}:y_{j}\geq0~\forall j,\sum_{j=1}^{N-2}y_{j}\leq
1-u\right\}  $. After a simple change-of-variable we obtain the conditional
density%
\[
\frac{\Gamma\left(  \alpha_{1}+\cdots+\alpha_{N-1}\right)  }{\Gamma\left(
\alpha_{1}\right)  \cdots\Gamma\left(  \alpha_{N-1}\right)  \left(
1-u\right)  ^{\beta-1}}\prod_{j=1}^{N-2}y_{j}^{\alpha_{j}-1},
\]
where $\beta=\sum_{j=1}^{N-1}\alpha_{j},y\in E_{u}$ and $y_{N-1}%
:=1-u-\sum_{j=1}^{N-2}y_{j}$. Now we can take the limit $u\rightarrow0_{+}$
and obtain the obvious Dirichlet distribution on $T_{N-2}$.

By applying this general result to the Cholesky factor (where $y_{j}=x_{j}%
^{2}$ or $\left\vert x_{j}\right\vert ^{2}$ and $N=10,16$ for $\alpha=\frac
{1}{2},1$ respectively) we find that the formula for the integral of a
monomial with respect to the ($x_{4}=0$)-conditional density is very similar
to the general one
\begin{equation}
\left\langle p\right\rangle =\frac{\left(  1+3\alpha\right)  _{n_{1}}\left(
1+2\alpha\right)  _{n_{2}}\left(  1+\alpha\right)  _{n_{3}}\left(  0\right)
_{n_{4}}}{\left(  3+12\alpha\right)  _{n}}\prod_{j=5}^{10}\left(
\alpha\right)  _{n_{j}},\label{densx40}%
\end{equation}
with the same rules for $n_{i}$ as before; the effect of the term $\left(
0\right)  _{n_{4}}$ is that $\left\langle p\right\rangle =0$ for any monomial
having $x_{4}$ as a factor; note $\left(  0\right)  _{0}=1$ (the empty product).

We proceed to the main results (conjectures): there is a natural decomposition%
\[
\left\vert \rho^{PT}\right\vert =f_{1}\left(  x^{\prime}\right)  +x_{4}%
^{2}f_{2}\left(  x^{\prime}\right)  +\left\vert \rho\right\vert ,
\]
where $x^{\prime}$ omits $x_{4}$ (that is, $x^{\prime}=\left(  x_{1}%
,\cdots,x_{3},x_{5},\cdots,x_{10}\right)  $). In the previous section there is
a formula for $\left\langle \left\vert \rho\right\vert ^{k}\left(  \left\vert
\rho^{PT}\right\vert -\left\vert \rho\right\vert \right)  ^{n}\right\rangle
=\left\langle \left\vert \rho\right\vert ^{k}\left(  f_{1}\left(  x^{\prime
}\right)  +x_{4}^{2}f_{2}\left(  x^{\prime}\right)  \right)  ^{n}\right\rangle
$. Note $f_{1}\left(  x^{\prime}\right)  =\left\vert \rho_{0}^{PT}\right\vert
$ when $\rho\in\Omega_{0}^{C}$. It is desired to find $\left\langle
f_{1}\left(  x^{\prime}\right)  ^{n}\right\rangle $, namely the $n^{th}$
moment of $\left\vert \rho_{0}^{PT}\right\vert $, for the conditional density
(\ref{densx40}). The key step is to consider%
\[
\left\langle \left(  x_{1}x_{2}x_{3}\right)  ^{2k}x_{4}^{2k-2}\left(
f_{1}\left(  x^{\prime}\right)  +x_{4}^{2}f_{2}\left(  x^{\prime}\right)
\right)  ^{n}\right\rangle ,
\]
for $k=1,2,3,\ldots$ ($k>0$ is required for integrability). Consider the
integrals of monomials (notations as in (\ref{monom})):%
\begin{gather}
\left\langle \left(  x_{1}x_{2}x_{3}x_{4}\right)  ^{2k}p\right\rangle
=\frac{\left(  1+3\alpha\right)  _{k}\left(  1+2\alpha\right)  _{k}\left(
1+\alpha\right)  _{k}\left(  1\right)  _{k}}{\left(  4+12\alpha\right)  _{4k}%
}\label{det1}\\
\times\frac{\left(  k+1+3\alpha\right)  _{n_{1}}\left(  k+1+2\alpha\right)
_{n_{2}}\left(  k+1+\alpha\right)  _{n_{3}}\left(  k+1\right)  _{n_{4}}%
}{\left(  4+4k+12\alpha\right)  _{n}}\prod_{j=5}^{10}\left(  \alpha\right)
_{n_{j}},\label{det2}%
\end{gather}
and%
\begin{gather}
\left\langle \left(  x_{1}x_{2}x_{3}\right)  ^{2k}x_{4}^{2k-2}p\right\rangle
=\frac{\left(  1+3\alpha\right)  _{k}\left(  1+2\alpha\right)  _{k}\left(
1+\alpha\right)  _{k}\left(  1\right)  _{k-1}}{\left(  4+12\alpha\right)
_{4k-1}}\label{det3}\\
\times\frac{\left(  k+1+3\alpha\right)  _{n_{1}}\left(  k+1+2\alpha\right)
_{n_{2}}\left(  k+1+\alpha\right)  _{n_{3}}\left(  k\right)  _{n_{4}}}{\left(
3+4k+12\alpha\right)  _{n}}\prod_{j=5}^{10}\left(  \alpha\right)  _{n_{j}%
}.\label{det4}%
\end{gather}
Observe that the quantity on the right side of (\ref{det1}) is $\left\langle
\left\vert \rho\right\vert ^{k}\right\rangle $ and the quantity on the right
side of (\ref{det3}) is $\left\langle \left(  x_{1}x_{2}x_{3}\right)
^{2k}x_{4}^{2k-2}\right\rangle $. \newline
{\bf{Conjecture}} \newline
For $k=1,2,3,\ldots,\alpha=\frac{1}{2}$ or $1$%
\begin{gather*}
\left\langle \left(  x_{1}x_{2}x_{3}\right)  ^{2k}x_{4}^{2k-2}\left(
\left\vert \rho^{PT}\right\vert -\left\vert \rho\right\vert \right)
^{n}\right\rangle =\left\langle \left(  x_{1}x_{2}x_{3}\right)  ^{2k}%
x_{4}^{2k-2}\right\rangle \\
\times\frac{\left(  2+4k+10\alpha+3n\right)  \left(  3+4k+12\alpha+4n\right)
}{\left(  2+4k+10\alpha+4n\right)  \left(  3+4k+12\alpha\right)  }F_{2}\left(
n,k\right)  ,
\end{gather*}
where $F_{2}$ is given in (\ref{F2p}) and (\ref{F2fact}).

The conjectured formula can be written as
\begin{align*}
& \frac{\left\langle \left(  x_{1}x_{2}x_{3}\right)  ^{2k}x_{4}^{2k-2}\left(
f_{1}\left(  x^{\prime}\right)  +x_{4}^{2}f_{2}\left(  x^{\prime}\right)
\right)  ^{n}\right\rangle }{\left\langle \left(  x_{1}x_{2}x_{3}\right)
^{2k}x_{4}^{2k-2}\right\rangle }\\
& =\left(  -1\right)  ^{n}\frac{\left(  \alpha\right)  _{n}\left(
\alpha+\frac{1}{2}\right)  _{n}\left(  n+2k+1+5\alpha\right)  _{n}\left(
2+10\alpha+4k+3n\right)  }{2^{4n}\left(  k+3\alpha+\frac{3}{2}\right)
_{n}\left(  2k+6\alpha+\frac{3}{2}\right)  _{2n}\left(  2+10\alpha
+4k+2n\right)  }\\
& \times~_{4}F_{3}\left(
\genfrac{}{}{0pt}{}{-\frac{n}{2},\frac{1-n}{2},k+1+\alpha,k+1+2\alpha
}{1-n-\alpha,\frac{1}{2}-n-\alpha,n+2k+2+5\alpha}%
;1\right)  .
\end{align*}

\begin{corollary}
The ($x_{4}=0$)-conditional expectation%
\begin{align*}
\left\langle f_{1}\left(  x^{\prime}\right)  ^{n}\right\rangle  &
=\left\langle \left\vert \rho_{0}^{PT}\right\vert ^{n}\right\rangle \\
& =\left(  -1\right)  ^{n}\frac{\left(  \alpha\right)  _{n}\left(
\alpha+\frac{1}{2}\right)  _{n}\left(  n+1+5\alpha\right)  _{n}\left(
2+10\alpha+3n\right)  }{2^{4n}\left(  3\alpha+\frac{3}{2}\right)  _{n}\left(
6\alpha+\frac{3}{2}\right)  _{2n}\left(  2+10\alpha+2n\right)  }\\
& \times~_{4}F_{3}\left(
\genfrac{}{}{0pt}{}{-\frac{n}{2},\frac{1-n}{2},1+\alpha,1+2\alpha
}{1-n-\alpha,\frac{1}{2}-n-\alpha,n+2+5\alpha}%
;1\right)  .
\end{align*}

\end{corollary}

\begin{proof}
From (\ref{det3}) and (\ref{det4}), for a monomial $p$ we have%
\[
\frac{\left\langle \left(  x_{1}x_{2}x_{3}\right)  ^{2k}x_{4}^{2k-2}%
p\right\rangle }{\left\langle \left(  x_{1}x_{2}x_{3}\right)  ^{2k}%
x_{4}^{2k-2}\right\rangle }=\frac{\left(  k+1+3\alpha\right)  _{n_{1}}\left(
k+1+2\alpha\right)  _{n_{2}}\left(  k+1+\alpha\right)  _{n_{3}}\left(
k\right)  _{n_{4}}}{\left(  3+4k+12\alpha\right)  _{n}}\prod_{j=5}^{10}\left(
\alpha\right)  _{n_{j}},
\]
with the same conditions on $\left\{  n_{i}\right\}  $ as before. The limit of
this expression as $k\rightarrow0_{+}$ equals the ($x_{4}=0$)-conditional
expectation $\left\langle p\right\rangle $. From the property of $\left(
0\right)  _{n_{4}}$ it follows that $\left\langle \left(  f_{1}\left(
x^{\prime}\right)  +x_{4}^{2}f_{2}\left(  x^{\prime}\right)  \right)
^{n}\right\rangle =\left\langle f_{1}\left(  x^{\prime}\right)  ^{n}%
\right\rangle $.
\end{proof}

The conjecture has been verified (by symbolic computation) for $n=1,2,3,4$,
$\alpha=\frac{1}{2}$ and $\alpha=1$. (One notes that the $n=4,\alpha=1$
computation involves around 8000 monomials with a nonzero integral, and
roughly 4 million monomials with zero integral. Each of these monomials is of
degree 32 in 16 variables.)\\
\bigskip
{\bf{Appendix C: Equal probabilities}}

For a probability density supported on a bounded interval $I$ let $\mu
_{n}\left[  f\right]  :=\int_{I}x^{n}f\left(  x\right)  dx,$ for
$n=0,1,2,\ldots$. The following is the probability density for the random
variable $XY$ where the density function of $X$ is $f\left(  x\right)  $, the
density of $Y$ is $\gamma y^{\gamma-1}$ on $0<y<1,$ and $X,Y$ are independent.
It is obvious that $\Pr\left\{  XY>0\right\}  =\Pr\left\{  X>0\right\}  $.

\begin{definition}
Suppose $f\left(  x\right)  $ is a probability density function supported on
$\left[  a,b\right]  $ with $a<0<b$, and $\gamma>0$, then%
\[
M\left(  \gamma\right)  f\left(  x\right)  :=\left\{
\begin{array}
[c]{c}%
\gamma x^{\gamma-1}\int_{x}^{b}f\left(  t\right)  t^{-\gamma}dt,~0<x<b\\
\gamma\left\vert x\right\vert ^{\gamma-1}\int_{a}^{x}f\left(  t\right)
\left\vert t\right\vert ^{-\gamma}dt,~a<x<0.
\end{array}
\right.
\]

\end{definition}

\begin{proposition}
\label{Mgmf}Suppose $f,\gamma$ are as in the definition, and $n=0,1,2,3,\ldots
$ then $M\left(  \gamma\right)  f$ is a probability density such that%
\begin{align*}
\int_{0}^{b}x^{n}M\left(  \gamma\right)  f\left(  x\right)  dx  &
=\frac{\gamma}{\gamma+n}\int_{0}^{b}x^{n}f\left(  x\right)  dx,\\
\int_{a}^{0}x^{n}M\left(  \gamma\right)  f\left(  x\right)  dx  &
=\frac{\gamma}{\gamma+n}\int_{a}^{0}x^{n}f\left(  x\right)  dx,\\
\mu_{n}\left[  M\left(  \gamma\right)  f\right]    & =\frac{\gamma}{\gamma
+n}\mu_{n}\left[  f\right]  .
\end{align*}

\end{proposition}

\begin{proof}
In the following iterated integral change the order of integration and then
evaluate the inner integral:%
\begin{align*}
\int_{0}^{b}x^{n}M\left(  \gamma\right)  f\left(  x\right)  dx  & =\gamma
\int_{0}^{b}x^{n+\gamma-1}dx\int_{x}^{b}f\left(  t\right)  t^{-\gamma
}dt=\gamma\int_{0}^{b}f\left(  t\right)  t^{-\gamma}dt\int_{0}^{t}%
x^{n+\gamma-1}dx\\
& =\frac{\gamma}{\gamma+n}\int_{0}^{b}f\left(  t\right)  t^{-\gamma+n+\gamma
}dt.
\end{align*}
Similarly%
\begin{align*}
\int_{a}^{0}x^{n}M\left(  \gamma\right)  f\left(  x\right)  dx  & =\gamma
\int_{a}^{0}x^{n}\left\vert x\right\vert ^{\gamma-1}dx\int_{a}^{x}f\left(
t\right)  \left\vert t\right\vert ^{-\gamma}dt\\
& =\gamma\int_{a}^{0}f\left(  t\right)  \left\vert t\right\vert ^{-\gamma
}dt\int_{t}^{0}x^{n}\left\vert x\right\vert ^{\gamma-1}dx;
\end{align*}
in the inner integral change the variable $x=-u$ and then integrate to obtain
$\left(  -1\right)  ^{n}\dfrac{\left(  -t\right)  ^{n+\gamma}}{\gamma
+n}=\dfrac{t^{n}\left\vert t\right\vert ^{\gamma}}{\gamma+n}$ for $a<t<0$.
This establishes the first two equations and the sum of the two shows $\mu
_{n}\left[  M\left(  \gamma\right)  f\right]  =\frac{\gamma}{\gamma+n}\mu
_{n}\left[  f\right]  $.
\end{proof}

\begin{proposition}
\label{Mdelg}Suppose $f$ is a density function on $\left[  a,b\right]  $ with
$a<0<b,$ $\gamma>0$ and $0<\delta<1$, then $g\left(  x\right)  :=\delta
f\left(  x\right)  +\left(  1-\delta\right)  M\left(  \gamma\right)  f\left(
x\right)  $ is a density function on $\left[  a,b\right]  $ and%
\begin{align*}
\int_{0}^{b}x^{n}g\left(  x\right)  dx  & =\frac{\gamma+\delta n}{\gamma
+n}\int_{0}^{b}x^{n}f\left(  x\right)  dx,\\
\int_{a}^{0}x^{n}g\left(  x\right)  dx  & =\frac{\gamma+\delta n}{\gamma
+n}\int_{a}^{0}x^{n}f\left(  x\right)  dx,\\
\mu_{n}\left[  g\right]    & =\frac{\gamma+\delta n}{\gamma+n}\mu_{n}\left[
f\right]  .
\end{align*}

\end{proposition}

\begin{proof}
It is clear that $g$ is a density. The other claims follow from the previous
proposition. For example%
\begin{align*}
\int_{0}^{b}x^{n}g\left(  x\right)  dx  & =\delta\int_{0}^{b}x^{n}f\left(
x\right)  dx+\left(  1-\delta\right)  \int_{0}^{b}x^{n}M\left(  \gamma\right)
f\left(  x\right)  dx\\
& =\left\{  \delta+\left(  1-\delta\right)  \frac{\gamma}{\gamma+n}\right\}
\int_{0}^{b}x^{n}f\left(  x\right)  dx\\
& =\frac{\gamma+\delta n}{\gamma+n}\int_{0}^{b}x^{n}f\left(  x\right)  dx.
\end{align*}

\end{proof}

Now fix $\alpha$ ($=\frac{1}{2},1$ for the real and complex cases) and denote
the density function of $\left(  \left\vert \rho^{PT}\right\vert -\left\vert
\rho\right\vert \right)  $ by $f\left(  x\right)  $, supported on $\left[
-\frac{1}{16},\frac{1}{432}\right]  $, or for arbitrary $\alpha>0$ take the
density function whose $n^{th}$ moments are given by $F_{2}\left(  n,0\right)
$ (see equations \ref{F2p} and \ref{F2fact} ); also denote the density
function of $\left\vert \rho_{0}^{PT}\right\vert $ in the minimally degenerate
setting by $g\left(  x\right)  $ (more generally the density whose moments are
given by $\dfrac{\left(  2+10\alpha+3n\right)  \left(  3+12\alpha+4n\right)
}{\left(  2+10\alpha+4n\right)  \left(  3+12\alpha\right)  }F_{2}\left(
n,0\right)  $).

Let $U_{1}$ denote the random variable with density $M\left(  3\alpha+\frac
{3}{4}\right)  g\left(  x\right)  $, then by Proposition \ref{Mgmf} the range
of $U_{1}$ is $\left[  -\frac{1}{16},\frac{1}{432}\right]  $ and%
\begin{align*}
\mu_{n}\left[  M\left(  3\alpha+\frac{3}{4}\right)  g\right]    &
=\frac{3\alpha+\frac{3}{4}}{3\alpha+\frac{3}{4}+n}\mu_{n}\left[  g\right]
=\frac{3+12\alpha}{3+12\alpha+4n}\mu_{n}\left[  g\right]  ,~n=0,1,2,3,\ldots\\
\Pr\left\{  U_{1}>0\right\}    & =\Pr\left\{  \left\vert \rho_{0}%
^{PT}\right\vert >0\right\}  .
\end{align*}
Let $U_{2}$ denote the random variable with density $h\left(  x\right)
:=\frac{3}{4}f\left(  x\right)  +$ $\frac{1}{4}M\left(  \frac{1+5\alpha}%
{2}\right)  f\left(  x\right)  $, then by Proposition \ref{Mdelg} the range of
$U_{2}$ is $\left[  -\frac{1}{16},\frac{1}{432}\right]  $ and%
\begin{align*}
\mu_{n}\left[  h\right]    & =\frac{\frac{1+5\alpha}{2}+\frac{3n}{4}}%
{\frac{1+5\alpha}{2}+n}\mu_{n}\left[  f\right]  =\frac{2+10\alpha
+3n}{2+10\alpha+4n}\mu_{n}\left[  f\right]  ,~n=0,1,2,3,\ldots,\\
\Pr\left\{  U_{2}>0\right\}    & =\Pr\left\{  \left\vert \rho^{PT}\right\vert
-\left\vert \rho\right\vert >0\right\}  .
\end{align*}

\begin{proposition}
Suppose $\alpha>0$ then $\Pr\left\{  \left\vert \rho^{PT}\right\vert
-\left\vert \rho\right\vert >0\right\}  =\Pr\left\{  \left\vert \rho_{0}%
^{PT}\right\vert >0\right\}  $.
\end{proposition}

\begin{proof}
From the conjecture it follows that $\mu_{n}\left[  M\left(  3\alpha+\frac
{3}{4}\right)  g\right]  =\mu_{n}\left[  h\right]  $ for all $n$. By the
uniqueness of moments (on bounded intervals) $U_{1}$ and $U_{2}$ have the same
density and $\Pr\left\{  \left\vert \rho^{PT}\right\vert -\left\vert
\rho\right\vert >0\right\}  =\Pr\left\{  U_{2}>0\right\}  =\Pr\left\{
\left\vert \rho_{0}\right\vert ^{PT}>0\right\}  $.
\end{proof}
\bigskip
{\bf{Appendix D: Specific moments for $6 \times 6$ Hilbert-Schmidt and $4 \times 4$ Bures scenarios}}   \\
{{\bf{$6 \times 6$ Hilbert-Schmidt moments}} \\
{{\bf{$n = 1$}} \\
\begin{equation} \label{rebitretrit}
\frac{\left\langle |\rho|^k (|\rho^{PT}| -|\rho|) \right\rangle_{rebit-retrit/HS}}{\left\langle |\rho|^k \right\rangle_{rebit-retrit/HS}}= 
 -\frac{5 (k+2) (k+3) (2 k+7)}{96 (k+4) (3 k+11) (3 k+13) (6 k+23) (6 k+25)},
\end{equation}
\begin{equation} \label{qubitqutrit}
\frac{\left\langle |\rho|^k (|\rho^{PT}|-|\rho|)  \right\rangle_{qubit-qutrit/HS}}{\left\langle |\rho|^k \right\rangle_{qubit-qutrit/HS}}=  -\frac{(k+3) (k+5) (2 k+11)}{3 (2 k+13) (3 k+19) (3 k+20) (6 k+37) (6 k+41)},
\end{equation}
and
\begin{equation} \label{quaterbitqutrit}
\frac{\left\langle |\rho|^k (|\rho^{PT}|-|\rho|)  \right\rangle_{quaterbit-quatertrit/HS}}{\left\langle |\rho|^k \right\rangle_{quaterbit-quatertrit/HS}}= 
-\frac{5 (k (k (3 k+70)+521)+1194)}{6 (2 k+23) (3 k+34) (3 k+35) (6 k+67) (6
   k+71)}.
\end{equation}
{{\bf{$n = 2$}} \\
\begin{equation} \label{rebitretrit2}
\frac{\left\langle |\rho|^k (|\rho^{PT}| -|\rho|)^2 \right\rangle_{rebit-retrit/HS}}{\left\langle |\rho|^k \right\rangle_{rebit-retrit/HS}}= 
\end{equation}
\begin{displaymath}
\frac{5 \left(24 k^6+900 k^5+10974 k^4+63561 k^3+193602 k^2+302033
   k+192132\right)}{82944 (k+5) (3 k+11) (3 k+13) (3 k+14) (3 k+16) (6 k+23) (6
   k+25) (6 k+29) (6 k+31)}
\end{displaymath}
and
\begin{equation} \label{qubitqutrit2}
\frac{\left\langle |\rho|^k (|\rho^{PT}| -|\rho|)^2 \right\rangle_{qubit-qutrit/HS}}{\left\langle |\rho|^k \right\rangle_{qubit-qutrit/HS}}= 
\end{equation}
\begin{displaymath}
\frac{k^6+69 k^5+1315 k^4+11475 k^3+51964 k^2+119856 k+112680}{108 (2 k+15) (3
   k+19) (3 k+20) (3 k+22) (3 k+23) (6 k+37) (6 k+41) (6 k+43) (6 k+47)}.
\end{displaymath}
{{\bf{$n = 3$}} \\
\begin{equation} \label{rebitretrit3}
\frac{\left\langle |\rho|^k (|\rho^{PT}| -|\rho|)^3 \right\rangle_{rebit-retrit/HS}}{\left\langle |\rho|^k \right\rangle_{rebit-retrit/HS}}= \frac{-35 A}{663552B}
\end{equation}
where
\begin{displaymath}
A=24 k^8+976 k^7+16438 k^6+152052 k^5+852799 k^4+2987211 k^3+
\end{displaymath}
\begin{displaymath}
6400915 k^2+7669535
   k+3920730
\end{displaymath}
and
\begin{displaymath}
B=(k+6) (3 k+11) (3 k+13) (3 k+14) (3 k+16) (3 k+17)
\end{displaymath}
\begin{displaymath}
 (3 k+19) (6 k+23) (6 k+25) (6
   k+29) (6 k+31) (6 k+35) (6 k+37).
\end{displaymath}
{{\bf{$4 \times 4$ Bures moments}} \\
{{\bf{$n = 1$}} \\
\begin{equation} \label{BuresTwoRebits}
\frac{\left\langle |\rho|^k (|\rho^{PT}|-|\rho|)  \right\rangle_{two-rebit/Bures}}{\left\langle |\rho|^k \right\rangle_{two-rebit/Bures}}= -\frac{k (16 k (k+7)+245)+173}{128 (k+2)^2 (2 k+3) (2 k+5) (2 k+7)}
\end{equation}
and
\begin{equation} \label{BuresTwoQubits}
\frac{\left\langle |\rho|^k (|\rho^{PT}|-|\rho|)  \right\rangle_{two-qubit/Bures}}{\left\langle |\rho|^k \right\rangle_{two-qubit/Bures}}= -\frac{3 (2 k (16 k (k+10)+499)+1005)}{128 (k+3) (k+4) (k+5) (4 k+9) (4 k+11)}.
\end{equation}
{{\bf{$n = 2$}} \\
\begin{equation} \label{BuresTwoRebits2}
\frac{\left\langle |\rho|^k (|\rho^{PT}|-|\rho|)^2  \right\rangle_{two-rebit/Bures}}{\left\langle |\rho|^k \right\rangle_{two-rebit/Bures}}= \frac{A}{4194304 B}
\end{equation}
where
\begin{displaymath}
A=4182016 k^{10}+87822336 k^9+745901568 k^8+3257689088 k^7+8000700112
   k^6+16462195504 k^5+
\end{displaymath}
\begin{displaymath}
65217922488 k^4+254319857272 k^3+570485963797
   k^2+660408583199 k+311220769578\end{displaymath}
and
\begin{displaymath}
B=(k+2)^2 (k+3)^2 (2 k+3) (2 k+5)^2 (2 k+7) (2 k+9) (2 k+11).
\end{displaymath}

For each of the nine  (six Hilbert-Schmidt and three Bures) results above, one can perform a transformation $k \rightarrow (k+r)$, where $r$ is a simple rational number, so
that the coefficient of the second-highest power in the numerator becomes zero. (As throughout our paper $\alpha =\frac{1}{2}$,  1 and 2 denote real, complex and 
quaternionic scenarios, respectively.)
For the Hilbert-Schmidt cases, we have $r =\frac{17}{6}$  $[n = 1, \alpha = \frac{1}{2}]$;
$r = \frac{9}{2}$  $ [n = 1, \alpha =1]$; $r = \frac{70}{9}$  $ [n = 1, \alpha =2]$;  $r = \frac{25}{4}$  $ [n = 2, \alpha =\frac{1}{2}]$; $r= \frac{23}{2} $  $[n =2, \alpha = 1]$; and $r = \frac{61}{12}$  $[n =3, \alpha= \frac{1}{2}]$. For the three Bures cases: $r= \frac{7}{3}$ $[n =1, \alpha = \frac{1}{2}]$; $r =\frac{10}{3}$ $[n=1, \alpha =1]$; and $r = \frac{21}{10}$ $[n =2, \alpha = \frac{1}{2}]$.

A formula that fits the first two equations ((\ref{rebitretrit}) and (\ref{qubitqutrit}))
in this appendix is
\begin{equation}
-\frac{4 \alpha  (\alpha +1) (\alpha +2) (2 \alpha +k+1) (4 \alpha +k+1) (5
   \alpha +k+1) (8 \alpha +2 k+3)}{(30 \alpha +6 k+6) (30 \alpha +6 k+7) (30
   \alpha +6 k+8) (30 \alpha +6 k+9) (30 \alpha +6 k+10) (30 \alpha +6 k+11)}.
\end{equation}
However, for $\alpha =2$, this formula yields
\begin{equation}
-\frac{4 (k+5) (k+9) (2 k+19)}{3 (2 k+23) (3 k+34) (3 k+35) (6 k+67) (6 k+71)},
\end{equation}
and not (\ref{quaterbitqutrit}).
\section{Appendix E. A symmetry property of separable states}
Let $\Omega$ denote the set of $2N\times2N$ Hermitian matrices, with entries
in $\mathbb{R},\mathbb{C}$, or $\mathbb{H}$, such that $M\in\Omega$ implies:

\begin{itemize}
\item $M_{ii}\geq0$ for $1\leq i\leq2N$, and $\sum_{i=1}^{2N}M_{ii}=1$;

\item $\left\vert M_{ij}\right\vert \leq1$ for $1\leq i,j\leq2N.$
\end{itemize}

Consider $\Omega$ as a compact (closed bounded) subset of $\mathbb{R}^{d}$
where $d=N\left(  2N+1\right)  $, $4N^{2}$, $2N\left(  4N-1\right)  $ for the
fields $\mathbb{R},\mathbb{C}$, $\mathbb{H}$ respectively, and furnish
$\Omega$ with the standard Euclidean (Lebesgue) measure $\mu$. This is
equivalent to the Hilbert-Schmidt measure.

Let $\sigma$ denote the operation of partial transposition. As an action on
the subset $\Omega$ of $\mathbb{R}^{d}$ it permutes some coordinates and
changes the sign of some other coordinates (for example $M_{12}$ is replaced
by $M_{21}=\overline{M_{12}}$, that is $\left(  \operatorname{Re}%
M_{12},\operatorname{Im}M_{12}\right)  \rightarrow\left(  \operatorname{Re}%
M_{12},-\operatorname{Im}M_{12}\right)  $). Thus $\sigma$ is an isometry (a
congruence relation) and preserves measure.

Let $P:=\left\{  M\in\Omega:M\geq0\right\}  $ (positive semi-definite).
Consider (for $M\in P$)%
\[
\Pr\left\{  \sigma M\geq0\right\}  =\frac{\mu\left(  P\cap\sigma P\right)
}{\mu\left(  P\right)  }.
\]
Now suppose $M\in P\cap\sigma P$ what can be deduced about $\Pr\left\{
\det\left(  \sigma M\right)  \geq\det M\right\}  $? Let%
\begin{align*}
E_{+}  & =\left\{  M\in P\cap\sigma P:\det\left(  \sigma M\right)  >\det
M\right\}  ,\\
E_{0}  & =\left\{  M\in P\cap\sigma P:\det\left(  \sigma M\right)  =\det
M\right\}  ,\\
E_{-}  & =\left\{  M\in P\cap\sigma P:\det\left(  \sigma M\right)  <\det
M\right\}  .
\end{align*}
Now $\sigma E_{+}=E_{-}$, $\sigma E_{-}=E_{+}$ and $\sigma E_{0}=E_{0}$ and
$P\cap\sigma P$ is the disjoint union of $E_{+},E_{0},E_{-}$ and $\mu\left(
P\cap\sigma P\right)  =\mu\left(  E_{+}\right)  +\mu\left(  E_{-}\right)
+\mu\left(  E_{0}\right)  $. The set $E_{0}$ is of lower dimension in
$\mathbb{R}^{d}$, thus $\mu\left(  E_{0}\right)  =0$. By the
measure-preserving property of $\sigma$ it follows that $\mu\left(
E_{-}\right)  =\mu\left(  \sigma E_{+}\right)  =\mu\left(  E_{+}\right)  $ and
hence $\mu\left(  E_{+}\right)  =\frac{1}{2}\mu\left(  P\cap\sigma P\right)
$, and
\[
\Pr\left\{  \det\left(  \sigma M\right)  \geq\det M\right\}  =\frac{\mu\left(
E_{+}\right)  }{\mu\left(  P\right)  }=\frac{1}{2}\frac{\mu\left(  P\cap\sigma
P\right)  }{\mu\left(  P\right)  }=\frac{1}{2}\Pr\left\{  \sigma
M\geq0\right\}  .
\]
As is known the condition $\sigma M\geq0$ is equivalent to $\det\left(  \sigma
M\right)  \geq0$ when $N=2,$ thus $M\in E_{+}$ is equivalent to $M\in P$ and
$\det\left(  \sigma M\right)  >\det M$.

\begin{acknowledgments}
PBS expresses appreciation to the Kavli Institute for Theoretical
Physics (KITP) for computational support in this research.
\end{acknowledgments}

\end{document}